\documentclass[11pt]{amsart}
\usepackage[english]{babel}
\usepackage{physics}
\usepackage{braket}
\usepackage{dsfont}
\usepackage{xcolor}
\usepackage{graphicx}
\usepackage{tikz}
\usepackage{amssymb}
\usepackage{amsmath}
\usepackage{mathrsfs}

\usepackage{polski}
\usepackage{verbatim}
 
\usetikzlibrary{matrix}

\usepackage[margin=1in]{geometry}
\usepackage[pagebackref, colorlinks = true, linkcolor = blue, urlcolor  = blue, citecolor = red]{hyperref}

\renewcommand{\epsilon}{\varepsilon}
\renewcommand{\phi}{\varphi}

\let\Tr\relax
\DeclareMathOperator{\Tr}{Tr}

\newcommand{\can}{\mathbf{can}}

\newtheorem{thm}{Theorem}[section]
\newtheorem*{thm*}{Theorem}
\newtheorem{lem}[thm]{Lemma}
\newtheorem{cor}[thm]{Corollary}

\newtheorem{defi}[thm]{Definition}
\newtheorem{prop}[thm]{Proposition}
\newtheorem{remark}[thm]{Remark}

\newtheorem{conjecture}[thm]{Conjecture}

\begin{document}

\title[A Fisher information-based incompatibility criterion for quantum channels]{A Fisher information-based incompatibility criterion\\for quantum channels}

\author{Qing-Hua Zhang}
\email{2190501022@cnu.edu.cn}
\address{School of Mathematical Sciences, Capital Normal University, Beijing 100048, China and Laboratoire de Physique Th\'eorique, Universit\'e de Toulouse, CNRS, UPS, France}

\author{Ion Nechita}
\email{ion.nechita@univ-tlse3.fr}
\address{Laboratoire de Physique Th\'eorique, Universit\'e de Toulouse, CNRS, UPS, France}

\begin{abstract} 
We introduce a new incompatibility criterion for quantum channels, based on the notion of (quantum) Fisher information. Our construction is based on a similar criterion for quantum measurements put forward by H.~Zhu. We then study the power of the incompatibility criterion in different scenarios. Firstly, we prove the first analytical conditions for the incompatibility of two Schur channels. Then, we study the incompatibility structure of a tuple of depolarizing channels, comparing the newly introduced criterion with the known results from asymmetric quantum cloning. 
\end{abstract}

\maketitle

\tableofcontents

\section{Introduction}
The impossibility of simultaneous realizations of two quantum operations is one of the fundamental features of quantum theory \cite{Heisenberg1927, Bohr1928}. Two of the most famous incarnations of this principle are the \emph{Heisenberg uncertainty principle} (the position and momentum of a quantum particle can not be measured simultaneously \cite{Heisenberg1985Uber}) and the \emph{no-cloning theorem} (there is no physical operation producing two identical copies of an unknown, arbitrary, quantum state \cite{dieks1982communication,wootters1982single}). In general, two (or more) quantum operations, such as measurements, channels, or instruments, are called \emph{compatible} if they can be seen as marginals of a common operation; if there is no physical operation having the original ones as marginals, they are called \emph{incompatible}. As quantum theory is built on Hilbert space, general quantum measurements are consider as the positive operator-valued measures (POVMs). In quantum information theory, there are many applications of the notion of incompatibility, such as the robustness of entanglement \cite{Vidal1999Robustness,Steiner2003Generalized}, the robustness of measurement incompatibility \cite{Uola2015One,Uola2019Quantifying,Designolle2019Incompatibility}, quantum non-locality \cite{Wolf2009Measurements,Bene2018Measurement}, quantum steering \cite{Quintino2014Joint,Uola2015One}, quantum state discrimination \cite{Carmeli2019Quantum,Skrzypczyk2019All,Mori2020Operational}, quantum resource theory \cite{Chitambar2019Quantum}, and quantum cryptography \cite{Brunner2014Brunner}. 

In modern formalism of quantum theory, the most general description of physical transformations of quantum states is in terms of \emph{quantum channels} \cite{nielsen2010Quantum,Busch1995Operational}. The concept of incompatibility of quantum channels has been proposed in terms of the input-output devices \cite{Heinosaari2016An,Heinosaari2017Incompatibility}. In \cite{Heinosaari2017Incompatibility}, the authors show that the definition of \emph{incompatibility of quantum channels} is a natural generalization of joint measurability of quantum observables. There exists a large body of work dealing with this notion from various points of view \cite{Kuramochi2018Quantum,Carmeli2019Witnessing,Mori2020Operational,Teiko2015Incompatibility}. Generally speaking, deciding whether a given family of quantum operations is compatible can be formulated as a \emph{semidefinite program} \cite{boyd2004convex}. However, the size of the program grows \emph{exponentially} with the number of operations considered. Hence, this method can be computationally prohibitive even for small system size (such as qubits) when the number of systems is moderately large. To cope with this dimensionality problem, (in-)compatibility criteria have been introduced; these are conditions which are only necessary, or sufficient, for the compatibility of the given tuple of channels. As it is the case with quantum measurements \cite{Heinosaari2016An}, there exist much more compatibility criteria \cite{girard2021jordan} than incompatibility criteria. 

\medskip

In this work, we introduce a new \emph{incompatibility criterion for quantum channels}, based on the notion of (quantum) Fisher information. Our criterion is based on a similar condition put forward by H.~Zhu \cite{Zhu2015Information,zhu2016universal} in the case of quantum measurements. 

\medskip

After introducing the necessary background on Fisher information and quantum channel compatibility (Sections \ref{sec:Fisher} and \ref{sec:quantum-channels}), we put forward the new incompatibility criterion in Section \ref{sec:criterion}. The statement of the main result of the paper can be found in Theorem \ref{thm:incompatibility-criterion}. We then apply this result to study, for the first time, the incompatibility Schur channels, an important class of quantum operations with wide-ranging applications; see Section \ref{sec:Schur}. In the final two sections of the paper, we introduce different compatibility structures for assemblages of quantum channels (Section \ref{sec:assemblages}) and we study them in the case of generalized depolarizing channels (Section \ref{sec:depolarizing}). 

\section{Classical and quantum Fisher information}\label{sec:Fisher}

 Consider a family of probability distributions $\{p(x |\theta), \, x\in R\}$ parametrized by $\theta$. A central research direction in statistics is to estimate the accuracy of the value of parameter $\theta$ by observing outcomes $x$ sampled from the distributions. Recall that the (classical) \emph{Fisher information} of the model is defined as
\begin{equation*}
 I(\theta):=\sum_x{ p(x |\theta)\left(\frac{\partial \log p(x|\theta)}{\partial \theta}\right)^2}.   
\end{equation*}
When an estimator $\hat{\theta}(x)$ of the parameter $\theta$ is unbiased, the inverse of the classical Fisher information gives a lower bound on the \emph{mean square error} (MSE) of the estimator, which is the well-known \emph{Cram\'er-Rao bound} \cite{Rao1992Information,Cramer2016Mathematical}.
The notion of the classical Fisher information plays a significant role in the geometrical approach to statistics \cite{Amari1985Differential,nielsen2020elementary} and the information-theory approach to physics \cite{Frieden1998Physics}.

In the multiple-parameter scenario, when $\theta$ is a vector, the classical Fisher information is considered as a matrix form, which is a real symmetric matrix with matrix elements \cite{fisher1925Theory,Yuen1973Multiple,Braunstein1994Statistical}:
\begin{equation*}
    I_{ij}(\theta):= \sum_x{ p(x |\theta)\frac{\partial \log p(x |\theta)}{\partial \theta_i}\frac{\partial \log p(x |\theta)}{\partial \theta_j}}.
\end{equation*}

In quantum parameter estimation scenario, we may perform the quantum positive operator valued measurement (POVM) on a quantum state which depends on a parameter, to extract the parameter information. Consider a quantum measurement $\mathbf{M}=\{M_x\geq 0, \sum_x M_x=I_d\}$ acting on the states $\rho(\theta) \in \mathcal L(H_d)$. The parameterized probability of outcomes $x$ of the measurement is $p(x|\theta)=\Tr[\rho(\theta) M_x]$. The corresponding measurement-induced Fisher information $I_{\mathbf{M}}(\theta)$ is then given by
\begin{equation*}
    I_{\mathbf{M}}(\theta)=\sum_x{ p(x |\theta)\Tr\left(\frac{\partial \log \rho(\theta)}{\partial \theta}M_x\right)^2}  . 
\end{equation*}

The \emph{quantum Fisher information} of the model $\rho(\theta)$ is defined as \cite{Holevo1982Probabilistic}
$$ J(\theta):=\Tr[\rho(\theta)L(\theta)^2],$$
where the symmetric logarithmic derivative (SLD) operators $L(\theta)$ for the parameter $\theta$ are determined implicitly by
$$ \frac{{\rm d}\rho(\theta)}{\theta}=\frac{1}{2}[\rho(\theta)L(\theta)+L(\theta)\rho(\theta)].$$

In contrast with classical Cram\'er-Rao bound, the inverse of quantum Fisher information is also a lower bound for the MSE of unbiased estimator, which is called the \emph{quantum Cram\'er-Rao bound} \cite{Holevo1982Probabilistic}.

In quantum multiple-parameter estimation scenario, both the measure-induced Fisher information and quantum Fisher information are matrices
\begin{align}
    \label{eq:def-I} 
    I_{{\mathbf{M}},ij}(\theta) &= \sum_x{ p(x |\theta)\Tr\left(\frac{\partial \log \rho(\theta)}{\partial \theta_i}M_x\Big){\rm tr}\Big(\frac{\partial \log \rho(\theta)}{\partial \theta_j}M_x\right)},\\
    \label{eq:def-J} J_{ij}(\theta)&=\frac{1}{2}\Tr\{\rho(\theta)[L_i(\theta)L_j(\theta)+L_j(\theta)L_i(\theta)]\},
\end{align}
where $L_k$ is the SLD operator corresponding with $\theta_k$.
The measurement-induced Fisher information resembles the classical correlations, while the quantum Fisher information resembles the quantum mutual information. From the Braunstein-Caves theorem \cite{Braunstein1994Statistical}, the quantum Fisher information, independent of measurement, is an upper bound of the measurement-induced Fisher information in the positive semidefinite order for matrices:
$$I_{\mathbf{M}}(\theta)\leq J(\theta).$$
In this work, we shall consider another relationship proposed by Gill and Massar \cite{Gill2000State} for any $d-$dimensional quantum system:
\begin{equation}\label{GMineq}
    \Tr[J^{-1}(\theta)I_{\mathbf M}(\theta)]\leq d-1.
\end{equation}
This inequality was the main ingredient in the incompatibility criterion invented by Zhu \cite{Zhu2015Information}, which lies at the foundation of our incompatibility criterion for quantum channels. 

\section{Compatibility of quantum channels}\label{sec:quantum-channels}

In this section we review the basic definitions about quantum channel compatibility. 

Let $H_d$ and $H_D$ be Hilbert spaces and $\mathcal{L}(H_d)$ denote the family of linear operators on $H_d$. In the Schr\"odinger picture, a \emph{quantum channel} is defined as a linear map $\Phi:\mathcal{L}(H_d)\rightarrow \mathcal{L}(H_D)$ having the following two properties: 
\begin{itemize}
    \item \emph{complete positivity}: for any dimension $k \geq 1$, the linear map $\mathrm{id}_k \otimes \Phi:\mathcal{L}(\mathbb C^k \otimes H_d)\rightarrow \mathcal{L}(\mathbb C^k \otimes H_D)$ is a positive operator;
    \item \emph{trace preservation}: for all operators $X \in \mathcal L(H_d)$, $\Tr \Phi(X) = \Tr X$.
\end{itemize}
We say that quantum channels are trace preserving, completely positive (TPCP) maps. 
In this paper, we shall also consider the Heisenberg picture of quantum mechanics, where channels are seen as acting on observables instead of states. This amounts to considering the adjoint map $\Phi^*:\mathcal{L}(H_D)\rightarrow \mathcal{L}(H_d)$, where the adjoint is taken with respect to the Hilbert-Schmidt scalar product on the corresponding matrix spaces \cite{Heinosaari2017Incompatibility}:
\begin{equation*}
    \langle A,\Phi(\rho)\rangle=\langle \Phi^*(A),\rho\rangle,
\end{equation*}
where $\rho\in \mathcal{L}(H_d)$, $ A\in \mathcal{L}(H_D)$ and $\langle X,Y\rangle:=\Tr(X^* Y)$. 

We now recall the definition of the compatibility of quantum channels and refer the reader to the review \cite{Heinosaari2016An} for futher properties. 

\begin{defi}\label{def:channel-compatibility}
Consider two quantum channels $\Phi_1:\mathcal{L}(H_d)\rightarrow \mathcal{L}(H_{d_1})$ and $\Phi_2:\mathcal{L}(H_{d})\rightarrow \mathcal{L}(H_{d_2})$ having the same input space. The pair $(\Phi_1, \Phi_2)$ is said to be \emph{compatible}, if there exists a \emph{joint channel} $\Lambda:\mathcal{L}(H_d)\rightarrow \mathcal{L}(H_{d_1}\otimes H_{d_2})$ such that $\Phi_{1,2}$ are the \emph{marginals} of $\Lambda$:
$$\forall X \in \mathcal L(H_d), \quad \Phi_1(X) = \Tr_2 \Lambda(X) \quad \text{ and } \quad \Phi_2(X) = \Tr_1 \Lambda(X),$$
where $\Tr_{1,2}$ denote the partial trace operations in $\mathcal{L}(H_{d_1}\otimes H_{d_2}) \cong \mathcal L(H_{d_1}) \otimes \mathcal L(H_{d_2})$.

In the Heisenberg (dual) picture, the condition above reads
$$\forall A \in \mathcal L(H_{d_1}),\quad \Phi_1^*(A) = \Lambda^*(A\otimes I_{d_2}),\quad \text{ and } \quad \forall B \in \mathcal L(H_{d_2}),\quad \Phi_2^*(B) = \Lambda^*(I_{d_1}\otimes B).$$
The (in-)compatibility of more than two channels is defined in a similar manner. 
\end{defi}

As an example, let us consider the partially depolarizing channel, which is defined as:
\begin{equation}
    \Phi_t=t\cdot \mathrm{id}+(1-t)\Delta,\qquad 0\leq t\leq 1,
\end{equation}
with $\mathrm{id}(A)=A$ and $\Delta(A)=\Tr A)I_d/d$ for any operator $A$; these quantum channels will be discussed at length in Section \ref{sec:depolarizing}. From the \emph{no-cloning theorem} \cite{dieks1982communication,wootters1982single}, it follows that two copies of the identity channel, $(\mathrm{id}, \mathrm{id})$ are incompatible. On the other hand, the completely depolarizing channel $\Delta$ is compatible with any other channel. A question is the self-incompatibility of $\Phi_t$. It well known the channel $\Phi_t$ is self-compatible if $0\leq t\leq \frac{d+2}{2(d+1)}$ \cite{Werner1998Optimal,girard2021jordan}.  The necessary and sufficient condition for the compatibility of two different depolarizing channels $\Phi_{s}$ and $\Phi_{t}$ was shown in \cite{hashagen2017universal, Haapasalo2019Compatibility, nechita2021geometrical}:
\begin{equation}\label{eq:asymmetric-cloning}
    t+s-\frac{2}{d}\sqrt{(1-t)(1-s)}\leq 1.
\end{equation}

As previously discussed, quantum channel incompatibility is a key phenomenon in quantum theory, being at the heart of fundamental results in quantum information, such as the \emph{no-cloning theorem}. In order to measure the degree of incompatibility of a given set of quantum channels, several definitions of \emph{robustness of incompatibility} have been considered in the literature \cite{haapasalo2015robustness,Uola2019Quantifying,girard2021jordan}. In this section, we introduce a new such measure for a tuple of channels, which has the merit to take into consideration the asymmetry between the channels considered. A similar definition was considered in the case of POVMs in \cite{bluhm2018joint,bluhm2020compatibility}. We shall consider only channels acting on $\mathcal L(H_d)$, and we recall that $\Delta$ denotes the fully depolarizing channel $\Delta(X) = (\Tr X) I_d / d$.

\begin{defi}\label{def:Gamma}
Given an $N$-tuple of quantum channels $\mathbf \Phi := (\Phi_1, \Phi_2, \ldots, \Phi_N)$, define the \emph{compatibility region} of $\mathbf \Phi$ as
$$\Gamma_{\mathbf \Phi} := \Bigg\{s \in [0,1]^N \, : \, \text{ the channels } \Big[s_i \Phi_i + (1-s_i) \Delta\Big]_{i=1}^N \text{ are compatible}\Bigg\}.$$
\end{defi}

Note that the definition is relevant event in the case where the channels $\Phi_i$ are identical: $\Phi_i = \Phi$ for all $i \in [N]$, in which case we call $\Gamma_\Phi := \Gamma_{\mathbf \Phi}$ the \emph{self-compatibility} region (note that the dependence in $N$ is still present, since we are consider $N$ copies of the channel $\Phi$). 

The following result is a simple exercise. 

\begin{prop}
For any $N$-tuple of quantum channels $\mathbf \Phi := (\Phi_1, \Phi_2, \ldots, \Phi_N)$, the set $\Gamma_{\mathbf \Phi}$ is convex and closed (i.e.~a convex body). We have $0 \in \Gamma_{\mathbf \Phi}$, and, for all $i \in [N]$, 
$$e_i := (0, \ldots, 0, \underbrace{1}_{i\text{-th position}}, 0, \ldots, 0) \in \Gamma_{\mathbf \Phi}.$$
\end{prop}

\section{Channel incompatibility via POVM incompatibility}\label{sec:criterion}

This is the main section of our paper, where we put forward a new incompatibility criterion for quantum channels in Theorem \ref{thm:incompatibility-criterion}. Our criterion is based on an incompatibility criterion for quantum measurements (POVMs) introduced by H.~Zhu and his collaborators \cite{Zhu2015Information, zhu2016universal}.

Let us start by recalling the definition of compatibility (or joint measurability) of quantum measurements. First, recall that a \emph{quantum measurement} (or \emph{POVM}) is a $k$-tuple of operators $\mathbf A = (A_1, A_2, \ldots, A_k)$, having the following two properties: 
\begin{itemize}
    \item \emph{positivity}: the operators $A_1, \ldots, A_k \in \mathcal L(H_d)$ are positive semidefinite;
    \item \emph{normalization}: $\sum_{i=1}^k A_i = I_d$.
\end{itemize}
A POVM gives the most general form of a physical process which produces the probabilities given by the \emph{Born rule}: when measuring a quantum system described by a density matrix $\rho$, one obtains the result $i \in [k]$ with probability
$$\mathbb P[ \text{ outcome } i] = \Tr(\rho A_i).$$
Naturally, one can see a POVM $\mathbf A$ as a quantum-to-classical channel
$$\Phi_{\mathbf A}(X) = \sum_{i=1}^k \Tr(X A_i) \ketbra{i}{i},$$
where $\{\ket i\}_{i=1}^k$ denotes the canonical basis of $\mathbb C^k$ corresponding to the pointer states of the measurement apparatus. 

Whether two (or more) quantum measurement can be performed simultaneously is one of the crucial questions lying at the foundations of quantum theory \cite{Heisenberg1927, Bohr1928}. Mathematically, we have the following important definition (compare with Definition \ref{def:channel-compatibility}). 

\begin{defi}
Two POVMs ${\mathbf{A}}=\{A_i\}_{i \in [k]}$ and ${\mathbf{B}}=\{B_j\}_{k \in [l]}$ are said to be \emph{compatible} (or \emph{jointly measurable}) if there exists a third POVM $\mathbf C = \{C_{ij}\}_{(i,j) \in [k] \times [l]}$, called \emph{joint measurement}, such that 
\begin{align*}
    \forall i \in [k], \qquad &A_i = \sum_{j=1}^l C_{ij},\\
    \forall j \in [l], \qquad &B_j = \sum_{i=1}^k C_{ij}.
\end{align*}
Otherwise, the measurements $\mathbf{A}$ and $\mathbf{B}$ are called \emph{incompatible} \cite{Ali2009Commutative}. Compatibility of more than two measurements is defined similarly. 
\end{defi}

Quantum measurement (in-)compatibility has received a lot of attention in the literature, see e.g.~the excellent reviews \cite{Heinosaari2016An,guhne2021incompatible}, or the recent perspective on the problem focusing on the post-processing partial order \cite{heinosaari2022order}. Importantly for us, in \cite{Zhu2015Information}, H.~Zhu proposed a family of universal \emph{POVM incompatibility criteria} based on the classical Fisher information matrix. Assume a measurement ${\mathbf{C}}$ is joint measurement of  ${\mathbf{A}}_i$. According to the Fisher information data-processing inequality, the measurement-induced Fisher information matrix of ${\mathbf{A}}_i$ should not exceed that of ${\mathbf{C}}$, that is to say 
$$I_{{\mathbf{A}}_i}\leq I_{\mathbf{C}}$$
for all quantum states $\theta$ ($\theta$ is omitted in the formula above for convenience); the Fisher information matrix $I$ was defined in Eq.~\eqref{eq:def-I}. Define $\widetilde{I}_{{\mathbf{A}}_i}:=J^{-1/2}I_{{\mathbf{A}}_i}J^{-1/2}$ as the metric-adjusted Fisher information. The following inequality holds for compatible measurements based on Gill-Massar inequality \eqref{GMineq}: 
\begin{equation}\label{zhu1}
    \min\Big\{\Tr H \, : \, \ H \geq \widetilde{I}_{\mathbf A_i} \quad \forall i \in [N] \Big\}\leq d-1.
\end{equation}
Otherwise, the $N$-tuple of measurements $(\mathbf{A}_i)_{i \in [N]}$ is incompatible. When the parameter $\theta$ (the state around which we compute the Fisher information) corresponds to the maximally mixed state $\theta = I_d/d$, the inequality \eqref{zhu1} can be rephrased as the following proposition \cite{Zhu2015Information,zhu2016universal}. 

\begin{prop}\label{prop:criterion-POVM-incompatible}
For a set of $N$ measurements $\mathbf A = (\mathbf A_1, \mathbf A_2, \ldots, \mathbf A_N)$ on $\mathcal L(H_d)$, define the operators 
$$\forall i \in [N] \qquad G_{\mathbf{A}_i} := \sum_{s=1}^{k_i} |A_i(s)\rangle\langle A_i(s)|/[\Tr(A_i(s)] \in \mathcal L(H_d^{\otimes 2}),$$
where $A_i(1), A_i(2), \ldots, A_i(k_i)$ are the (non-zero) effects of the POVM $\mathbf A_i$, having $k_i$ outcomes. Consider now the quantity

\begin{equation}
\begin{aligned}
    \tau(\mathbf A) := \min \quad& \Tr H   \\
    s.t.\quad&H\geq G_{\mathbf A_i} \quad \forall i \in [N].
\end{aligned}
\end{equation}
If $\tau(\mathbf A)>d$, then the $N$-tuple of POVMs $\mathbf A = (\mathbf A_1, \mathbf A_2, \ldots, \mathbf A_N)$ is incompatible. 
\end{prop}

\begin{remark}
Note that the function $\tau (\mathbf{A})$ satisfies two basic requirements for a good measure of (in-)compatibility: monotonicity under coarse-graining and global unitary invariance.
\end{remark}

\begin{remark}\label{rem:G-larger-omega}
For any POVM $\mathbf A$, the associated matrix $G_{\mathbf A}$ is larger, in the positive semidefinite order, than the maximally entangled state
$$\omega := \frac 1 d \sum_{i,j=1}^d \ketbra{ii}{jj}.$$
This fact is a consequence of the important observation that the $\mathbf A \mapsto G_{\mathbf A}$ is an order morphism for the post-processing order of quantum measurements \cite{heinosaari2022order}, and $G_{\{I\}} = \omega$.
\end{remark}

A natural question is how to capture incompatibility of quantum channels using measurements. Let $\{|e_j\rangle\}$ and $\{|f_k\rangle\}$ be any sets of basis of Hilbert spaces $H_{d_1}$ and ${H_{d_2}}$, respectively. Motivated by the definition of incompatibility of quantum channels, we dedicate to research properties of the induced sets  $\{\Phi_1^*(|e_j\rangle\langle e_j|)\}$ and $\{\Phi_2^*(|f_k\rangle\langle f_k|)\}$. As we consider the quantum channel is trace-preserving, thus $\{\Phi_1^*(|e_j\rangle\langle e_j|)\}$ and $\{\Phi_2^*(|f_k\rangle\langle f_k|)\}$ can be regarded as POVMs \cite{Carmeli2019Witnessing}, that is to say,
\begin{equation}
    \sum_j \Phi_1^*(|e_j\rangle\langle e_j|)=  \sum_k \Phi_2^*(|f_k\rangle\langle f_k|)=I_d.
\end{equation}

\begin{lem}\label{lem:channel-POVM-compatible}
If $N$ quantum channels $\Phi_1,\Phi_2, \ldots, \Phi_N$ are compatible, then, for all orthonormal bases $\mathbf e^{(1)}, \mathbf e^{(2)}, \ldots, \mathbf e^{(N)}$ of $\mathbb C^d$, the corresponding POVMs 
$$ \mathbf A_s := \Big[ \Phi_s^*(|\mathbf e^{(s)}_i \rangle  \langle \mathbf e^{(s)}_i|) \Big]_{i=1}^d, \qquad \forall s \in [N]$$ 
are compatible. 
\end{lem}
\begin{proof}
Let $\Lambda$ be a joint channel for the compatible $N$-tuple $(\Phi_1, \Phi_2, \ldots, \Phi_N)$. Clearly, $\Lambda : \mathcal L(H_d) \to \mathcal L(H_d^{\otimes N})$, thus its adjoint is a unital, completely positive map
$$\Lambda^* : \mathcal L(H_d^{\otimes N}) \to \mathcal L(H_d).$$
Define operators $${\mathbf{B}}:= \Bigg[ \Lambda^* \Big( \bigotimes_{s=1}^N |\mathbf e^{(s)}_{i_s} \rangle \langle \mathbf e^{(s)}_{i_s} | 
\Big)\Bigg]_{i_1, \ldots, i_N \in [d]}.$$
From the fact that $\Lambda^*$ is a completely positive, unital map, we infer that $\mathbf B$ is a POVM (with $d^N$ outcomes). Let us now compute the marginals of this POVM. For some fixed $s \in [N]$ and $i_s \in [d]$, we have

\begin{equation}
\sum_{i_1, \ldots, i_{s-1}, i_{s+1}, \ldots, i_N \in [d]} B_{i_1 \cdots i_N} = \Lambda^*(I_d  \otimes \cdots \otimes |\mathbf e^{(s)}_{i_s} \rangle \langle \mathbf e^{(s)}_{i_s} | \otimes \cdots \otimes I_d) = \Phi_s^*(|\mathbf e^{(s)}_{i_s} \rangle \langle \mathbf e^{(s)}_{i_s} |) = \mathbf A_s(i_s),
\end{equation}
showing that the $s$-th marginal of $\mathbf B$ is $\mathbf A_s$. Thus ${\mathbf{B}}$ is a joint measurement of $\mathbf A_1, \mathbf A_2, \ldots, \mathbf A_N$, proving the claim.
\end{proof}

We leave open the reciprocal question, which we formulate as a conjecture (below for two channels, although the general version, for a $N$-tuple, can be easily stated). 
\begin{conjecture}
Consider two quantum channels $\Phi, \Psi: \mathcal L(H_d) \to \mathcal L(H_d)$ such that, for all orthonormal bases $\mathbf e = (e_1, \ldots, e_d)$, $\mathbf f = (f_1, \ldots, f_d)$ of $\mathbb C^d$, the POVMs
$$\Big[ \Phi^*(\ketbra{e_i}{e_i}) \Big]_{i=1}^d \quad \text{ and } \quad \Big[ \Psi^*(\ketbra{f_j}{f_j}) \Big]_{j=1}^d$$
are compatible. Then, $\Phi$ and $\Psi$ are compatible channels.
\end{conjecture}

We now turn to the main theoretical result of our paper: a criterion for quantum channel incompatibility. Informally, one can formulate it as follows: given an $N$-tuple of quantum channels, if one can find an $N$-tuple of orthonormal bases such that the corresponding quantum measurements are incompatible, then the original $N$-tuple of channels must also be incompatible. Our criterion is important since the are very few useful incompatibility criteria for channel incompatibility. On the other hand, there exist quite numerous incompatibility criteria for quantum measurements, so one can turn those into criteria for channels using Lemma \ref{lem:channel-POVM-compatible}. We introduce the following important notation: to a quantum channel $\Phi : \mathcal L(H_d) \to \mathcal L(H_d)$ and an orthonormal basis $\mathbf e = (e_i)_{i=1}^d$ of $\mathbb C^d$, we associate the ``$G$'' matrix
\begin{equation}\label{eq:def-G}
G_{\Phi, \mathbf e} := \sum_{i=1}^d \frac{\ketbra{\Phi^*(|e_i\rangle \langle e_i|)}{\Phi^*(|e_i\rangle \langle e_i|)}}{\Tr \Phi^*(|e_i\rangle \langle e_i|)},    
\end{equation}
which corresponds to the ``$G$'' matrix associated to the POVM
$$\Big[ \Phi^*(|e_i\rangle \langle e_i|) \Big]_{i=1}^d.$$

\begin{thm}\label{thm:incompatibility-criterion}
Let $\Phi_1,\Phi_2, \ldots, \Phi_N : \mathcal L(H_d) \to \mathcal L(H_d)$ be $N$ quantum channels. If there exists orthonormal bases $\mathbf e^{(1)}, \mathbf e^{(2)},  \ldots, \mathbf e^{(N)}$ of $\mathbb C^d$ such that the value of the semidefinite program
\begin{equation}\label{eq:SDP-incompatibility-criterion}
\begin{aligned}
    \min \quad& \Tr H   \\
    s.t.\quad&H\geq G_{\Phi_i, \mathbf e^{(i)}} \quad \forall i \in [N]
\end{aligned}
\end{equation}
is strictly larger than $d$, then the $n$-tuple of channels $\mathbf \Phi = (\Phi_1, \Phi_2, \ldots, \Phi_N)$ is incompatible. 
\end{thm}
\begin{proof}
The theorem follows directly from Proposition \ref{prop:criterion-POVM-incompatible} and Lemma \ref{lem:channel-POVM-compatible}. 
\end{proof}

\begin{remark}
If the quantum channel $\Phi$ is \emph{unital}, that is to say, $\Phi(I_d)=I_d$, the formula \eqref{eq:def-G} simplifies, in the sense that the denominator is trivial:  
\begin{equation}
    \Tr \Phi^*(\ketbra{e_i}{e_i})= \langle I_d, \Phi^*(\ketbra{e_i}{e_i}) \rangle = \langle \Phi(I_d), \ketbra{e_i}{e_i} \rangle = \langle I_d, \ketbra{e_i}{e_i} \rangle = \Tr \ketbra{e_i}{e_i} = 1.
\end{equation}
This will be the case for most of the examples we shall discuss in what follows. 
\end{remark}

It is important at this point to note that the incompatibility criterion we put forward in the result above is formulated as an SDP (semidefinite program). The usual way of formulating the compatibility of a tuple of quantum channels is also an SDP: one looks for a joint channel, a problem which can be formulated as an SDP thanks to the Choi formalism. However, let us compare the size of the SDPs: 
\begin{itemize}
    \item channel compatibility: the joint channel has a Choi matrix of size $d^{N+1}$
    \item incompatibility criterion from Theorem \ref{thm:incompatibility-criterion}: the variable $H$ has size $d^2$.
\end{itemize}
Note also that one has, in both cases, $N$ constraints of size $d^2$. So, we obtain a dramatic reduction in the size of the SDP, at the price of having only a necessary compatibility condition (i.e.~an incompatibility criterion). 

There is however a situation when the SDP \eqref{eq:SDP-incompatibility-criterion} simplifies, and can be analytically solved. This is when the matrices $G$ corresponding to the channel are orthogonal (up to the maximally entangled state $\omega$). We formalize this observation below. 

\begin{prop}\label{prop:orthogonal-incompatibility-criterion}
Consider $N$ quantum channels $\Phi_1,\Phi_2, \ldots, \Phi_N : \mathcal L(H_d) \to \mathcal L(H_d)$ and orthonormal bases $\mathbf e^{(1)}, \mathbf e^{(2)},  \ldots, \mathbf e^{(N)}$ such that, for all $i,j \in [N]$, $i \neq j$, 
$$G_{\Phi_i, \mathbf e^{(i)}} - \omega \perp G_{\Phi_j, \mathbf e^{(j)}} - \omega.$$
Then, the value of the SDP \eqref{eq:SDP-incompatibility-criterion} is  
$$1-N + \sum_{i=1}^N \Tr G_{\Phi_i, \mathbf e^{(i)}}.$$
\end{prop}

\begin{proof}
Taking into consideration Remark \ref{rem:G-larger-omega}, one can rewrite the SDP 
(\ref{eq:SDP-incompatibility-criterion}) by subtracting $\omega$ everywhere: 
\begin{equation}
\begin{aligned}
    1+\min \quad& \Tr \tilde H   \\
    s.t.\quad& \tilde H\geq G_{\Phi_i, \mathbf e^{(i)}} - \omega \quad \forall i \in [N]
\end{aligned}
\end{equation}
where $\tilde H = H - \omega$. Using the hypothesis, and noting that the matrices $G_{\Phi_i, \mathbf e^{(i)}} - \omega$ are all positive semidefinite, any feasible $\tilde H$ must satisfy
$$\tilde H \geq \sum_{i=1}^N G_{\Phi_i, \mathbf e^{(i)}} - \omega.$$
Hence, the optimal $\tilde H$ achieves equality above, and the conclusion follows. 
\end{proof}

This idea will be used in Sections \ref{sec:Schur} and \ref{sec:depolarizing} to obtain (analytical) incompatibility criteria for important classes of quantum channels. 

As an example, let us work out the ``$G$'' matrix for the identity channel $\mathrm{id}(X) = X$. 
\begin{equation}\label{eq:def-Z}
    G_{\mathrm{id}, \mathbf e} = \sum_{i=1}^d \Big | \ketbra{e_i}{e_i} \Big \rangle \Big \langle \ketbra{e_i}{e_i} \Big | = \sum_{i=1}^d \ketbra{e_i \otimes \bar e_i}{e_i \otimes \bar e_i} =: Z_{\mathbf e}.
\end{equation}
The matrix $Z_{\mathbf e}$ will play an important role in what follows. We gather some useful facts about it below. Recall that two orthonormal bases $\mathbf e, \mathbf f$ of $\mathbb C^d$ are called \emph{unbiased} if
$$\forall i,j \in [d], \qquad |\langle e_i, f_j \rangle| = \frac{1}{\sqrt d}.$$

\begin{lem}\label{lem:Z-and-omega}
For any orthonormal basis $\mathbf e$, we have 
$$\langle Z_{\mathbf e} , \omega \rangle = 1.$$
Moreover, if $\mathbf e$ and $\mathbf f$ are \emph{unbiased} orthonormal bases, then 
$$\langle Z_{\mathbf e} , Z_{\mathbf f} \rangle = 1.$$
\end{lem}

Let us close this section by mentioning how the matrices $G$ behave when mixing noise into a quantum channel $\Phi$. This property will be very useful in what follows when investigating the compatibility robustness of some classes of quantum channels.

\begin{lem}\label{lem:G-scaling}
Given a quantum channel $\Phi:\mathcal L(H_d) \to \mathcal L(H_d)$, consider its noisy version 
$$\Phi_t := t \Phi + (1-t) \Delta,$$
where $\Delta(X) = (\Tr X) I/d$ is the completely depolarizing channel and $t \in [0,1]$ is some parameter. Then, for any orthonormal basis $\mathbf e$, 
$$G_{\Phi_t, \mathbf e} = t^2 G_{\Phi, \mathbf e} + (1-t^2) \omega,$$
where $\omega$ is the maximally entangled state (note that $\omega = G_{\Delta, \mathbf e}$).
\end{lem}
\begin{proof}
This can be either proven directly using formula \eqref{eq:def-G} or by using the corresponding result for POVMs, see, e.g.~\cite[Proposition 5.3]{heinosaari2022order}.
\end{proof}

\section{Incompatibility of two Schur channels}\label{sec:Schur}

As a first application of our newly introduced incompatibility criterion for quantum channels, we consider \emph{Schur channels}. A Schur map is a linear map of the form 
$$\Sigma_B(X) = B \circ X,$$
where $B$ is a $d \times d$ complex matrix. The map $\Sigma_B$ is completely positive if and only if the matrix $B$ is positive semidefinite, and it is trace preserving if the diagonal of $B$ is the identity: $B_{ii} = 1$ for all $i$. If both conditions are satisfied, we call the map $\Sigma_B$ a Schur channel (sometimes also called a Schur multiplier), see \cite{paulsen2002schur, harris2018schur,watrous2018theory,singh2021diagonal}. Schur channels have received a lot of attention in operator algebra and quantum information theory, and they contain as examples the identity channel $\mathrm{id} = \Sigma_J$, where $J$ is the all 1s matrix, and the dephasing channel (the conditional expectation on the diagonal sub-algebra) $\mathrm{diag} = \Sigma_I$. 

For a Schur channel $\Sigma_B$, we have
$$G_{\Sigma_B, \mathbf e} = \ketbra{\bar B}{\bar B} \circ Z_{\mathbf e},$$
for any orthonormal basis $\mathbf e$ (recall the form of the matrix $Z$ from \eqref{eq:def-Z}). If $e$ is the canonical basis, we have 
$$G_{\Sigma_B,\can} = Z_\can.$$

Consider now a basis $\mathbf f$ which is unbiased with respect to the canonical basis; in other works, the elements of $\mathbf f$ form the columns of a Hadamard matrix $U$: $\ket{f_j} = U \ket j$ for all $j$. An important example of such a basis is the \emph{Fourier basis}: 
$$f_j(s) = \exp(2 \pi \mathrm{i} / d)^{js}, \qquad \forall j,s \in [d].$$

\begin{lem}\label{lem:Schur-orthogonal}
If $B,C$ are two positive semidefinite matrices with unit diagonal, and $\can$ and $\mathbf f$ are unbiased, then 
$$G_{\Sigma_B, \can} - \omega \perp G_{\Sigma_C, \mathbf f}- \omega.$$
\end{lem}
\begin{proof}
Expanding the scalar product and using Lemma \ref{lem:Z-and-omega}, we need to show that
$$\langle Z_\can, \ketbra{\bar C}{\bar C} \circ Z_{\mathbf f} \rangle = \langle \omega , \ketbra{\bar C}{\bar C} \circ Z_{\mathbf f} \rangle.$$
Let us work out the left-hand-side:
$$\langle Z_\can, \ketbra{\bar C}{\bar C} \circ Z_{\mathbf f} \rangle = \langle \ketbra{C}{C} \circ Z_\can, Z_{\mathbf f} \rangle = \langle Z_\can, Z_{\mathbf f} \rangle = 1,$$
where we have used Lemma \ref{lem:Z-and-omega} and the fact that $\ketbra{C}{C} \circ Z_\can = Z_\can$, which follows from the fact that $C$ has unit diagonal. The right-hand-side can be dealt with in the same manner. 
\end{proof}

For a $d \times d$ matrix $B$ with unit diagonal, define the real parameter $\beta(B)$ as follows:
\begin{equation}\label{eq:def-beta}
\beta(B) := \frac{1}{d-1}\left( \frac 1 d \sum_{i,j=1}^d |B_{ij}|^2 - 1\right) = \frac{1}{d-1} \sum_{i \neq j \in [d]} |B_{ij}|^2.
\end{equation}
Recall that the torus $\mathbb T^d$ is the set of vectors $b \in \mathbb C^d$ with $|b_i|=1$ for all $i = 1, \ldots, d$.

\begin{lem}
If $B$ is a $d \times d$ positive semidefinite matrix with unit diagonal, then 
$$0 \leq \beta(B) \leq 1,$$
with $\beta(B) = 0$ iff $B = I$ and $\beta(B) = 1$ iff $B = \ketbra{b}{b}$ for a vector $b \in \mathbb T^d$.
\end{lem}
\begin{proof}
The non-negativity of $\beta$, as well as the equality case, follows directly from the definition \eqref{eq:def-beta}. For the upper bound, use the ordering of the $1,2$-Schatten norms of $B$ to write
$$\sum_{i,j=1}^d |B_{ij}|^2 = \|B\|_2^2 \leq \|B\|_1^2 = (\Tr B)^2 = d^2,$$
proving the inequality. Equality holds if $B$ is rank one, which, together with the condition on the diagonal, proves that  $B = \ketbra{b}{b}$ for some vector $b \in \mathbb T^d$.
\end{proof}

We can now, using Theorem \ref{thm:incompatibility-criterion}, provide a new incompatibility criterion for Schur channels. 

\begin{thm}
Consider two positive semidefinite matrices $B,C$ with unit diagonal, and the corresponding depolarized Schur channels
\begin{align*}
    \Phi_s(X) &= s \Sigma_B(X) + (1-s) \Delta(X) = s B \circ X + (1-s) (\Tr X) \frac I d\\
    \Psi_t(X) &= t \Sigma_C(X) + (1-t) \Delta(X) = t C \circ X + (1-t) (\Tr X) \frac I d.
\end{align*}
If $s^2 + \beta(C) t^2 >1$, then the channels $\Phi_s$ and $\Psi_t$ are incompatible. We have thus an upper bound for the compatibility region from Definition \ref{def:Gamma}:
\begin{equation}\label{eq:Gamma-set-Schur}
  \Gamma_{\Phi,\Psi} \subseteq \{ (s,t) \in [0,1]^2 \, : \, s^2 + \beta(C) t^2 \leq 1 \text{ and } \beta(B) s^2 +  t^2 \leq 1\},  
\end{equation}
where $\Phi := \Phi_1$ and $\Psi := \Psi_1$.
\end{thm}
\begin{proof}
The proof is an application of Theorem \ref{thm:incompatibility-criterion}. To start, let us compute the ``$G$'' matrices associated to these channels, taking, respectively, the canonical basis $\can$, and any unbiased base $\mathbf f$ (e.g.~the Fourier basis); this choice is inspired by Lemma \ref{lem:Schur-orthogonal} and Proposition \ref{prop:orthogonal-incompatibility-criterion}. Applying these results, as well as the scaling Lemma \ref{lem:G-scaling}, we have 
\begin{align*}
    G_{\Phi_s, \can} &= s^2 G_{\Sigma_B, \can} + (1-s^2) \omega  = \omega + s^2(Z_\can - \omega)\\
    G_{\Psi_t, \mathbf f} &= t^2 G_{\Sigma_C, \mathbf f} + (1-t^2) \omega = \omega + t^2(\ketbra{\bar C}{\bar C} \circ Z_{\mathbf f} - \omega).
\end{align*}
Hence, the value of the SDP \eqref{eq:SDP-incompatibility-criterion} is given by (see Proposition \ref{prop:orthogonal-incompatibility-criterion})
$$1-2 + \Tr G_{\Phi_s, \can} + \Tr G_{\Psi_t, \mathbf f} = s^2(d-1) + 1-t^2 + t^2 \Tr[\ketbra{\bar C}{\bar C} \circ Z_{\mathbf f}].$$
We can evaluate 
$$ \Tr[\ketbra{\bar C}{\bar C} \circ Z_{\mathbf f}] =  \frac 1 d \sum_{i,j=1}^d |C_{ij}|^2,$$
and, using the parameter $\beta(C)$ from \eqref{eq:def-beta}, the incompatibility criterion reads
$$s^2 + \beta(C) t^2 >1,$$
which is the first claim. The second claim follows by swapping the roles of the unbiased bases $\can$ and $\mathbf f$.
\end{proof}

\begin{remark}
One can not easily generalize the result above to more that two Schur channels. This is due to the fact that one has to fix one of the bases in Theorem \ref{thm:incompatibility-criterion} to be the canonical basis. This is due to the fact that the Hadamard product used to define Schur channels is adapted to the canonical basis. We leave the generalization of the result (and method) above for three or more Schur channels open. 
\end{remark}

We compare in Figure \ref{fig:Schur} the criterion from the Theorem above with the actual incompatibility thresholds for some particular Schur channels, concluding that the incompatibility criterion is close to being exact. 

\begin{figure}
    \centering
    \includegraphics[width=.45\textwidth]{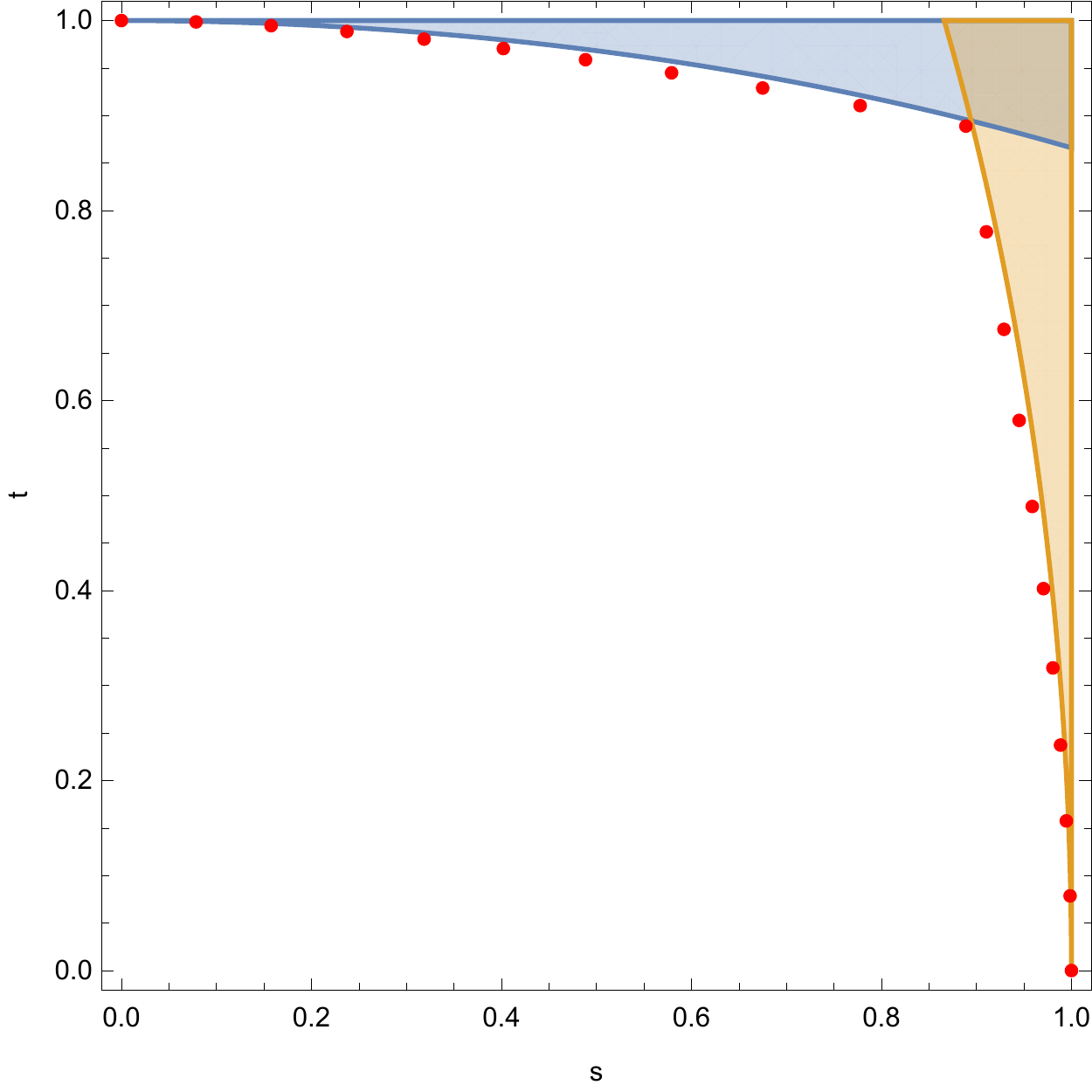} \qquad \includegraphics[width=.45\textwidth]{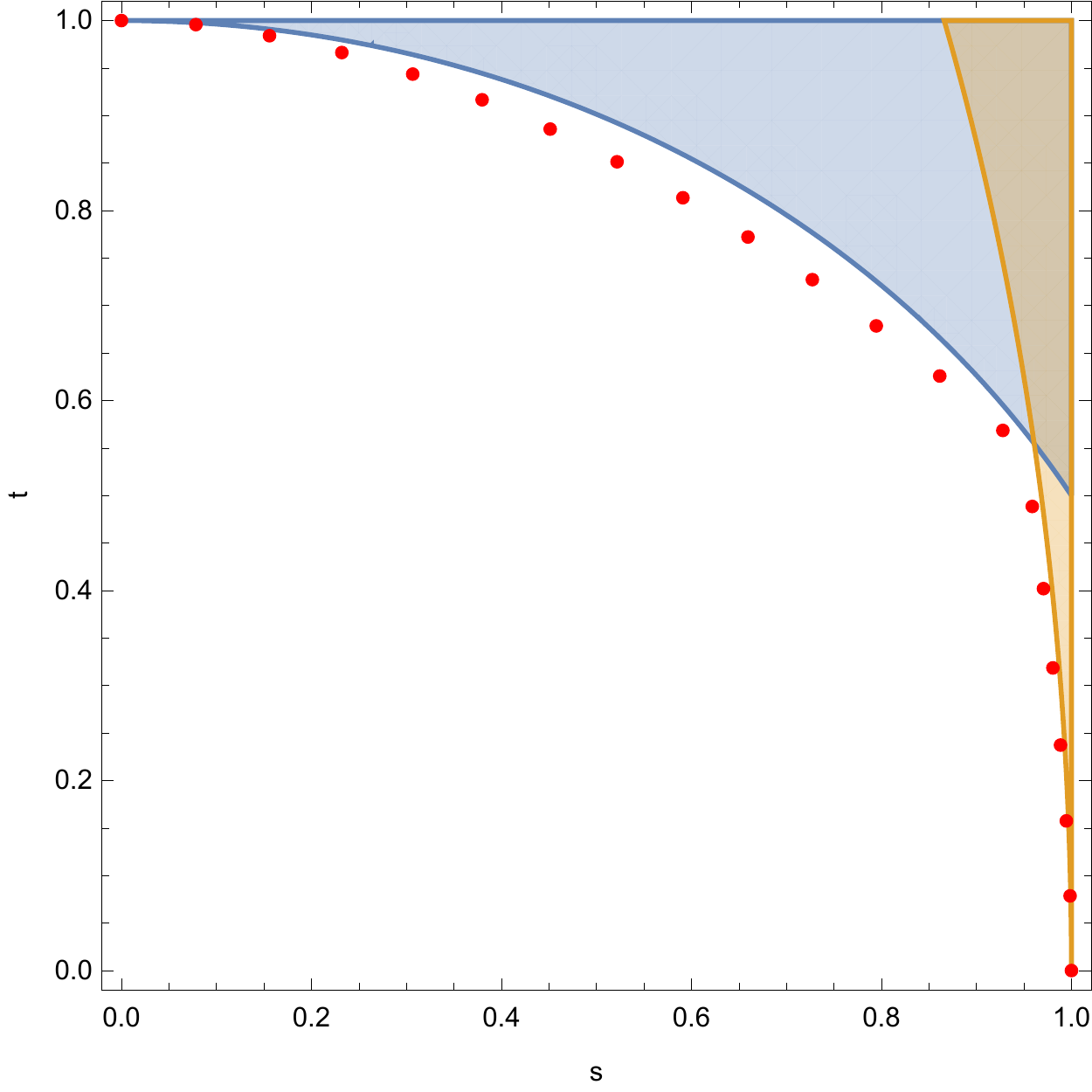} 
    \caption{The Fishere information based incompatibility criterion for Schur channels. In the left panel, we consider two noisy copies of the Schur channel corresponding to $B= \begin{bmatrix} 1 & 1/2 \\ 1/2 & 1 \end{bmatrix}$. In the right panel, we consider noisy versions of $\Sigma_B$ and $\Sigma_C$, where $C = \begin{bmatrix} 1 & \sqrt{3/4} \\ \sqrt{3/4} & 1 \end{bmatrix}$. Shaded regions correspond to the conditions from \eqref{eq:Gamma-set-Schur}, while the red dots correspond to the maximally compatible channels in the respective directions.}
    \label{fig:Schur}
\end{figure}

\section{Channel assemblages}\label{sec:assemblages}

The way in which several quantum measurement and quantum channels can be incompatible has been studied extensively in the literature \cite{liang2011specker,kunjwal2014quantum,yadavalli2020bell,sun2020note}. The kind of (in-)compatibility structures that can be found in Nature, and their relation to other important manifestations of non-locality (such as Bell inequality violations) is clearly a crucial question at the foundation of quantum theory.

Let $\{\Phi_i\}_{i=1}^{N}$ be a \emph{channel assemblage}, that is an $N$-tuple of quantum channels. If $\{\Phi_i\}_{i=1}^{N}$ are incompatible, there does not exist a joint quantum channel for \emph{all} the $N$ channels. However, a joint channel may exist when we consider certain subset of $\{\Phi_i\}_{i=1}^{N}$. In other words, some subsets of $\{\Phi_i\}_{i=1}^{N}$ may be compatible, even though the whole set is incompatible. Obviously, if the whole set of $N$ channels is compatible, then so is any subset: if $\Lambda$ is a joint channel for the $N$-tuple. then, for any subset $S \subseteq [N]$ of the channels, $\Lambda_S$, the marginal of $\Lambda$ corresponding to the output indices in $S$
$$\Lambda_S : \mathcal L(H_d) \to \mathcal L\left(\bigotimes_{i \in S} H_d^{(i)} \right)$$
is a joint channel for $\{\Phi_i\}_{i \in S}$; above, we identify the different copies of the output space $\mathcal L(H_d)$ by a superscript. Therefore, it is significant to classify the incompatibility of subsets for a given quantum channel assemblage. A \emph{$K$-subset} of $[N]$ is simply a subset $S \subseteq [N]$ of cardinality $|S| = K$.
 
 \begin{defi}\label{def:assemblage-incompatibility}
Consider a quantum channel assemblage $\mathbf \Phi = \{\Phi_i\}_{i=1}^{N}$ and $1 \leq K \leq N$ an integer. The $N$-tuple $\mathbf \Phi$ is called:
\begin{itemize}
    \item $(N,K)$-compatible if \emph{all} $K$-subsets of $\mathbf \Phi$ are compatible.
    \item $(N,K)$-incompatible if \emph{at least one} $K$-subset of $\mathbf \Phi$ is incompatible.
    \item $(N,K)$-strong incompatible if \emph{all} $K$-subsets of $\mathbf \Phi$ are incompatible.
    \item $(N,K+1)$-genuinely incompatible if it is $(N,K)$-compatible and $(N,K+1)$-incompatible.
     \item $(N,K+1)$-genuinely strong incompatible if it is $(N,K)$-compatible and $(N,K+1)$-strong incompatible.
\end{itemize}
 \end{defi}
 Note that the assemblage $\{\Phi_i\}_{i=1}^{N}$ is compatible if and only if it is $(N,N)$-compatible. The previous definition is strongly inspired by the one from \cite[Section 2]{sun2020note} in the case of POVMs. The incompatibility criterion from Theorem \ref{thm:incompatibility-criterion} can be readily adapted to the previous definition by considering subsets of the PSD constraints in \eqref{eq:SDP-incompatibility-criterion}. We restate it here for the convenience of the reader. We shall apply it in the next section for assemblages of depolarizing channels. 

\begin{thm}\label{thm:S-incompatibility-criterion}
Consider an assemblage $\mathbf \Phi = \{\Phi_i\}_{i=1}^{N}$ of quantum channels acting on $\mathcal L(H_d)$. For a $K$-subset $S$ of $[N]$, and $K$ orthonormal bases $\mathbf e = (\mathbf e^{(1)}, \mathbf e^{(2)},  \ldots, \mathbf e^{(K)})$ of $\mathbb C^d$, define the value of the following semidefinite program
\begin{equation}
\begin{aligned}
    \mathrm{val}(\mathbf \Phi, S, \mathbf e) := \min \quad& \Tr H   \\
    s.t.\quad&H\geq G_{\Phi_i, \mathbf e^{(i)}} \quad \forall i \in S.
\end{aligned}
\end{equation}
If there exists \emph{at least one} $S \in [N]$ and a $K$-tuple of orthonormal bases $\mathbf e$ such that $\mathrm{val}(\mathbf \Phi, S, \mathbf e) > d$, then the assemblage $\Phi$ is $(N,K)$-incompatible. Moreover, if for \emph{all} $K$-subsets $S \subseteq [N]$, there exists a $K$-tuple of bases $\mathbf e_S$ such that $\mathrm{val}(\mathbf \Phi, S, \mathbf e_S) > d$, the assemblage $\Phi$ is $(N,K)$-strong incompatible.
\end{thm}

\section{Assemblages of depolarizing channels}\label{sec:depolarizing}

In this section, we address the (in-)compatibility properties of an $N$-tuple of partially depolarizing channels, using the Fisher information based criterion from Theorem \ref{thm:incompatibility-criterion}. Recall that the \emph{partially depolarizing channel} is the linear map $\Phi_t : \mathcal L(H_d) \to \mathcal L(H_d)$ given by
\begin{equation}\label{eq:def-depolarizing-channel}
    \Phi_t=t\cdot \mathrm{id} +(1-t)\Delta,
\end{equation}
where $\mathrm{id}$ is the identity channel $\mathrm(X)=X$ and $\Delta$ is the fully depolarizing channel $\Delta(X)=(\Tr X)I/d$. The parameter $t \in [0,1]$ interpolates between the identity channel and fully depolarizing channel.

We shall study in this section the incompatibility of $N$ partially depolarizing channels \{$\Phi_{t_i}^i\}_{i=1}^N$, for some fixed parameters $t_1, t_2, \ldots, t_N \in [0,1]$, with the help of the criterion from Theorem \ref{thm:incompatibility-criterion}. To do so, let us first compute the ``$G$'' matrices of depolarizing channels, which are just noisy versions of the identity channel. Recall from Eq.~\eqref{eq:def-Z} that, for the identity channel, we have, for an arbitrary basis $\mathbf e$,
$$Z_{\mathbf{e}}=G_{\mathrm{id},\mathbf{e}}=\sum_{i=1}^d \ketbra{e_i \otimes \bar e_i}{e_i \otimes \bar e_i},$$
where $\bar e_i$ denotes the (entrywise) complex conjugate of the vector $e_i$. Hence, by Lemma \ref{lem:G-scaling}, we have 
$$ G_{\Phi_t,\mathbf{e}}=t^2 Z_{\mathbf{e}}+(1-t^2)\omega.$$

As in Section \ref{sec:Schur}, we are going to use the orthogonality of the ``$G$'' matrices in order to put forward analytical incompatibility criteria for depolarizing channels (see Proposition \ref{prop:orthogonal-incompatibility-criterion}). To do so, recall that the $Z_{\mathbf e}$ matrices have tractable scalar products for unbiased bases. As it turns out, mutually unbiased bases \cite{durt2010mutually} will play an important role in what follows. Let $D_d$ be the maximal cardinality of a set of mutually unbiased bases of $\mathbb C^d$. It is known that $3 \leq D_d \leq d+1$ \cite{wootters1989optimal,klappenecker2003constructions,combescure2007mutually}. The upper bounds is attained for all dimensions $d$ which are prime powers; whether it is always attained is an important open problem in quantum information theory, even the case $d=6$ being undecided. 

We now state the main result of this section, an incompatibility criterion for depolarizing channels. 

\begin{prop}\label{prop:depolarizing-compatible}
Let $N$ be an integer such that $N \leq D_d$, the maximal number of mutually unbiased bases of $\mathbb C^d$. Consider $N$ depolarizing channels $\Phi_{t_1}, \ldots, \Phi_{t_N}$, where $t_1, \ldots, t_N \in [0,1]$ are noise parameters. If 
\begin{equation}\label{eq:sum-t-squared}
    t_1^2 + t_2^2 + \cdots + t_N^2 >1
\end{equation}
then the $N$ depolarizing channels $\Phi_{t_i}$ are incompatible. 
\end{prop} 
\begin{proof}
Since the number of channels we consider is smaller than $D_d$, we can choose $N$ mutually unbiased bases $\mathbf e^{(1)}, \ldots, \mathbf e^{(N)}$. The SDP \eqref{eq:SDP-incompatibility-criterion} reads
$$
\begin{aligned}
    \min \quad& \Tr H   \\
    s.t.\quad&H\geq t_i^2 Z_{\mathbf e^{(i)}} + (1-t_i^2)\omega \quad \forall i \in [N].
\end{aligned}
$$
Proposition \ref{prop:orthogonal-incompatibility-criterion} applies, so the value of the SDP above is 
$$\min \Tr H= 1 + (d-1) \sum_{i=1}^N t_i^2.$$
Hence, if the condition \eqref{eq:sum-t-squared} holds, by Theorem \ref{thm:incompatibility-criterion} the $N$ quantum depolarizing channels $\Phi_{t_1}, \ldots, \Phi_{t_N}$ are incompatible. 
\end{proof}

As mentioned in the introduction, the compatibility of depolarizing channels is equivalent to approximate quantum cloning: how much noise one needs to add to $N$ copies of the identity channel to render them compatible. We present in Figure \ref{fig:cloning} the relative performance of the criterion from Proposition \ref{prop:depolarizing-compatible} with the true values of the noise parameters for $1 \to 2$ asymmetric approximate quantum cloning from Eq.~\eqref{eq:asymmetric-cloning}.

\begin{figure}
    \centering
    \includegraphics[width=.6\textwidth]{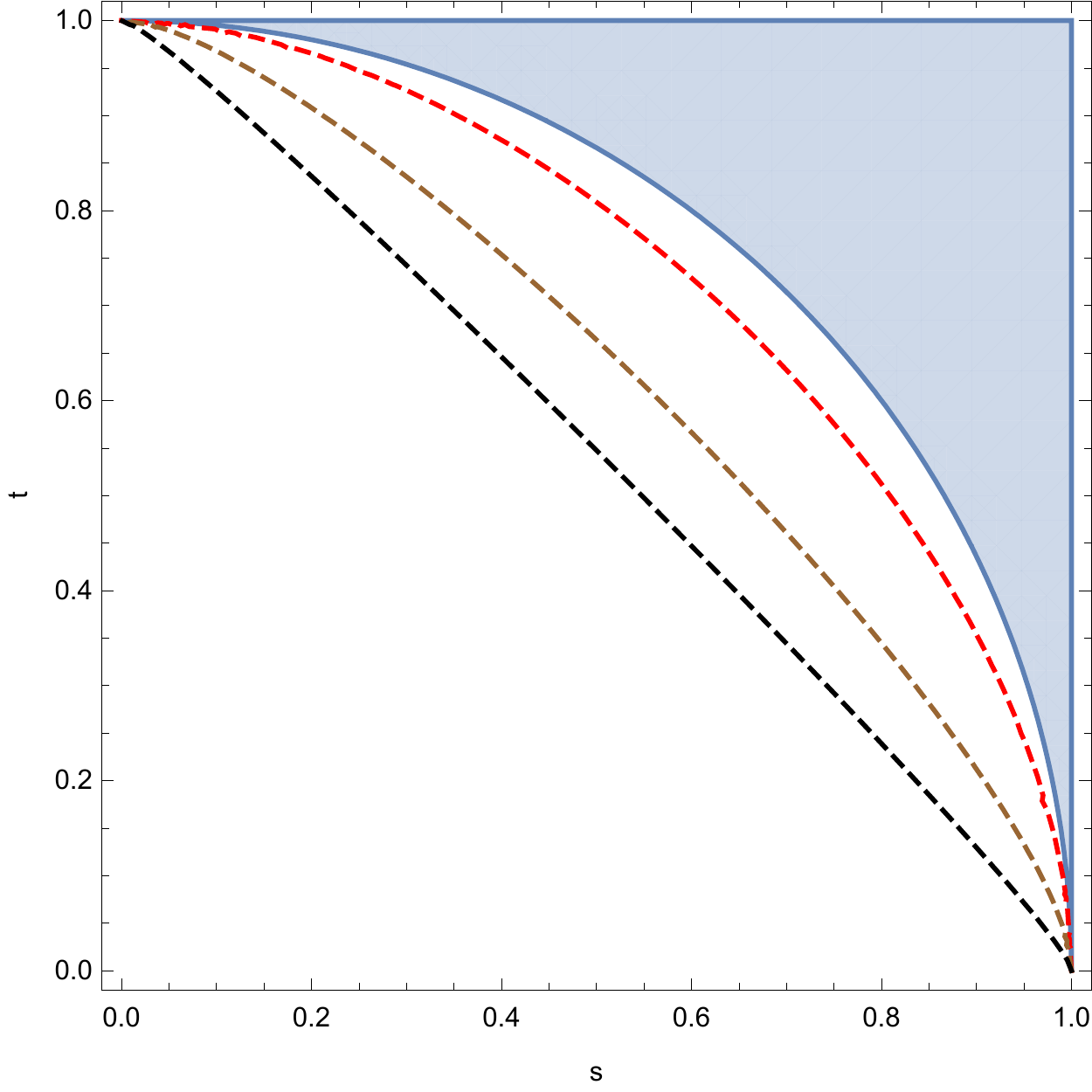}
    \caption{Comparing the incompatibility criterion from Proposition \ref{prop:depolarizing-compatible} (filled region) with the incompatibility thresholds from Eq.~\eqref{eq:asymmetric-cloning} (dashed curves) for different values of $d$: $d=2$ (red curve), $d=5$ (brown curve), $d=20$ (black curve).}
    \label{fig:cloning}
\end{figure}

We can specialize the result above to assemblages of depolarizing channels in the spirit of Definition \ref{def:assemblage-incompatibility}.
\begin{cor}
Consider $N$ partially depolarizing channels  \{$\Phi_{t_i}\}_{i=1}^N$ acting on $\mathcal L(H_d)$ and let $K \leq \min(N,D_d)$ be an integer. If there exist a subset $S \subseteq [N]$ of cardinality $K$ such that
$$\sum_{i\in S} t_i^2>1,$$
then the channels are $(N,K)$-incompatible. Moreover, if for \emph{every} subset $S \subseteq [N]$ of cardinality $K$ the condition above holds, the channels are $(N,K)$-strongly incompatible.
\end{cor} 

Note that in the statement above, we do not require that the number $N$ of channels must be smaller that the number of mutually unbiased bases in the corresponding Hilbert space; this is required only of the parameter $K$. This criterion might thus be useful in situations where one has a large number of channels. 

We end this section by a similar corollary, in the setting where the channels are identical. 
\begin{cor}
If $N,K$ are integers such that $K \leq \min(N, D_d)$, then the partially depolarizing channel $\Phi_t$ from Eq.~\eqref{eq:def-depolarizing-channel} is $(N,K)$-self-(strong) incompatible as soon as $t>1/\sqrt{K}$.
\end{cor} 

\bigskip

\noindent\textit{Acknowledgments.} The authors were supported by the ANR project \href{https://esquisses.math.cnrs.fr/}{ESQuisses}, grant number ANR-20-CE47-0014-01 and by the China Scholarship Council.

\bibliographystyle{alpha}
\bibliography{article}

\end{document}